\documentclass[aps,prl,10pt,a4paper,reprint,
superscriptaddress,
nofootinbib,twocolumn,longbibliography]{revtex4-1}
\usepackage{amsmath}
\usepackage{amsfonts}
\usepackage{amssymb}
\usepackage{graphicx}
\newcommand{\cH}{\mathcal{H}}
\newcommand{\ot}{\otimes}
\newcommand{\id}{\openone}
\newcommand{\ket}[1]{| #1 \rangle}

\newcommand{\ketbra}[2]{|#1\rangle\!\langle#2|}

\newcommand{\tr}{\operatorname{tr}}
\newcommand{\dbot}{\mathbin{\text{$\bot\mkern-8mu\bot$}}}
\newcommand{\indep}{\mathbin{\text{$\bot\mkern-8mu\bot$}}}

\newcommand{\st}{\mathrm{st}}
\newcommand{\Icomp}{I-compatible}
\newcommand{\bef}{\prec}

\newcommand{\nbef}{\nprec}
\newcommand{\naft}{\nsucc}
\usepackage[dvipsnames]{xcolor}
\usepackage{tikz,amsthm,nicefrac}
\usepackage[colorlinks=true, urlcolor=red, linkcolor=red, citecolor=blue, pdfborder={0 0 0}]{hyperref}
\definecolor{Md}{RGB}{232,238,248}   
\AtBeginDocument{%
    \newwrite\bibnotes
    \def\bibnotesext{Notes.bib}
    \immediate\openout\bibnotes=\jobname\bibnotesext
    \immediate\write\bibnotes{@CONTROL{REVTEX41Control}}
    \immediate\write\bibnotes{@CONTROL{%
    apsrev41Control,author="08",editor="1",pages="1",title="0",year="0"}}
     \if@filesw
     \immediate\write\@auxout{\string\citation{apsrev41Control}}%
    \fi
  }%
\usetikzlibrary{shapes.geometric,arrows.meta,positioning}
\tikzset{
  obs/.style={regular polygon, regular polygon sides=3, draw, thick,
              inner sep=0pt, minimum size=7mm, fill=GreenYellow},
  lat/.style={circle, draw, thick, inner sep=0pt, minimum size=6mm, fill=Melon},
  ed/.style={-{Latex[length=1.6mm]}, thick},
  edn/.style={-{Latex[length=1.6mm]}, thick, densely dashed}
}
\newtheorem{theorem}{Theorem}

\newtheorem{lemma}{Lemma}
\theoremstyle{definition}
\newtheorem{definition}{Definition}

\usepackage[nodayofweek]{datetime}

\begin{document}
\title{Spurious quantum correlations}
\author{Shashaank Khanna}
\email{shashaank.khanna@lis-lab.fr}
\affiliation{Department of Mathematics, University of York, Heslington, York, YO10 5DD, United Kingdom}
\affiliation{Aix-Marseille University, CNRS, LIS, Marseille, France}
\author{Matthew F. Pusey}
\email{matthew.pusey@york.ac.uk}
\affiliation{Department of Mathematics, University of York, Heslington, York, YO10 5DD, United Kingdom}
\author{Roger Colbeck}
\email{roger.colbeck@kcl.ac.uk}
\affiliation{Department of Mathematics, King's College London, Strand, London, WC2R 2LS, United Kingdom}

\date{$3^{\mathrm{rd}}$ September 2026}

\begin{abstract}
    In his seminal paper, Bell [Physics Physique Fizika {\bf 1}, 195 (1964)] considers the correlations that result from space-like separated measurements on a pair of entangled particles. He uses relativity theory to motivate the Bell causal structure, then shows the existence of quantum correlations that cannot be explained classically within this causal structure. Classical explanations of such quantum correlations are possible in alternative causal structures, for instance, those that allow superluminal causal influences, but, as shown in [New Journal of Physics {\bf 17} 033002 (2015)], all such alternative explanations require fine tuning (causation without correlation). Here we discuss the existence of spurious quantum correlations --- correlations that look quantum in one causal structure, but have a natural classical explanation in another. More precisely, there are causal structures that admit non-classical quantum correlations, but for which the same correlations have a classical explanation in another causal structure without fine tuning. The realisation in the other causal structure can also be achieved without breaking any natural constraints on the causal structure that follow from relativity theory. However, similarly to non-classical quantum correlations in the Bell causal structure, we find other causal structures with non-classical quantum correlations that do not have a classical causal explanation in any alternative causal structure without fine tuning.
\end{abstract}
\maketitle

\section{Introduction}
Bell's theorem~\cite{bell, bellnouvelle} establishes the existence of quantum correlations that cannot be explained by any local classical theory. To do so he considered a particular causal structure, and showed that measurements on a pair of entangled particles generate correlations violating a bound that holds whenever the shared particles are classical. Later, Wood and Spekkens~\cite{wood} considered whether other causal structures could realise these quantum correlations classically, and showed that doing so requires fine tuning, i.e., the presence of causation without correlation (see later for a precise definition). Such realisations would hence be somewhat conspiratorial, for instance, they might have superluminal causation while the observed correlations are no-signalling.

Taking a step back, we consider the task of finding possible causal explanations of a given phenomenon. More precisely, we consider a setup involving several variables and trying to infer the underlying causal relations. By performing repeated trials, and assuming an i.i.d.\ structure, we can get a close approximation of the distribution of the variables. For the purposes of the present paper, we consider the case in which we know the distribution exactly. A first feature we might look for is the set of conditional independences in the distribution, i.e., whether there are variables or sets of variables $X$, $Y$, $Z$ such that $P(X|YZ)=P(X|Z)$. Causal structures often directly imply sets of conditional independences, so any causal structure that implies a conditional independence not seen in the observed distribution can be eliminated. Furthermore, it is common to also eliminate causal structures in which almost all correlations that can be generated with that causal structure do not have a conditional independence that is present in the observations. This additional elimination corresponds to seeking causal structures for which the observed correlations are faithful (not fine tuned).

The IC* algorithm~\cite{pearl} provides a way to go from a distribution (or its conditional independences) to a set of possible faithful classical causal structures with the additional restriction that any latent common causes in the causal structure are pairwise (have at most two children). The justification for this restriction is the result that for a given set of observed conditional independences, if there exists a causal structure that generates them faithfully, then there also exists a causal structure with pairwise common causes that generates them faithfully~\cite{verma1993graphical}. This restriction does not cause problems for most of our results, but we need to go beyond it when discussing correlations that are not spurious. IC* also restricts latent variables to be parentless. This is because in a classical causal structure, for any latent variable that is not parentless, if a new causal structure is formed by removing the edges from its parents and adding a directed edge from each of its parents to each of its children, any distribution realisable classically in the original causal structure is classically realisable in the latter. Since our uses of IC* are to find classical causal structures, this second restriction will not cause issues in this work.

If only one causal structure results, then one should try to recreate the correlations within that structure, and if that is possible, this would appear to be an acceptable causal structure for the observed correlations\footnote{If not, then one is forced to reconsider the assumptions that go into IC*, e.g., it may be a sign that the correlations are fine tuned.}. In general, there are several causal structures that can recreate the conditional independences without fine tuning. In some cases, additional principles may be used to discard some of these. For instance, if each observed variable can be assigned a spacetime location, some causal structures may be rejected because they would have causation outside the future lightcone, in violation of relativity theory. E.g., if $B$ is observed outside the future lightcone of $A$, then it is reasonable to reject causal structures in which $A$ is a cause of $B$.\footnote{There is a bit of subtlety here in that we would need to be able to assign spacetime points or regions in which variables are first generated; in some cases this may not be easy to justify.}

As Bell's theorem shows, the sets of correlations that can be produced by a causal structure also depend on the theory. In this work we are interested in causal structures in which there are non-classical quantum correlations, but for which the same correlations can be explained classically in another causal structure. The alternative classical explanation should not be fine tuned, and it should also not be possible to eliminate it using the relativistic principle mentioned above. We refer to correlations with these properties as spurious quantum correlations because at first sight they appear quantum, but this may be because the wrong causal structure is being considered. From the result of Wood and Spekkens~\cite{wood} mentioned earlier, there are no spurious quantum correlations in the causal structure Bell considered.

To find spurious causal structures, we investigated the causal structures from~\cite{HLP_2014} that support non-classical correlations. We use the IC* algorithm to try to identify other causal structures that imply the same observed conditional independences. We find that for some causal structures, including our primary example $H_1$ in Figure~\ref{fig:spurious1}, or $H_2$ and $H_3$ in Figure~\ref{fig:spurious2}, all non-classical (quantum) correlations they can produce are spurious. However, there are other causal structures, like in Figure~\ref{fig:genuine}, with non-classical (quantum) correlations, for which no other causal structure faithfully allows the same conditional independences. Thus no non-classical (quantum) correlations in such causal structures are spurious.

\section{Classical and quantum correlations within a causal structure}
We first give a few definitions to help set up our analysis.

\begin{definition}[Causal Structure]
A \emph{causal structure} is a directed acyclic graph (DAG) over a set of latent (denoted here within a circle) and observed (denoted here within a triangle) nodes.
\end{definition}
 Directed edges in a causal structure indicate (the possibility of) direct causal influences. The observed nodes represent observable things and hence have an associated random variable\footnote{We use upper case to denote random variables and lower case to denote specific instances of those variables.}, while latent nodes are unobserved and the way they are modelled depends on the underlying theory (in this work we are mainly concerned with classical and quantum theory). Figure~\ref{fig:spurious1}(a) shows the main causal structure we consider in this paper.

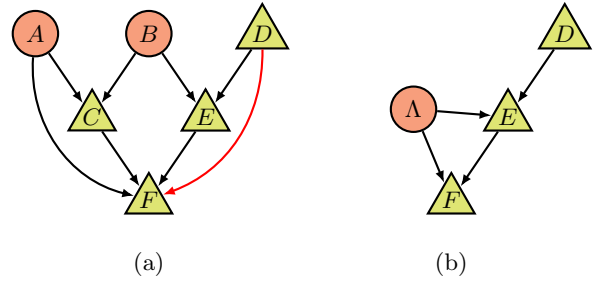
\begin{figure}[t]
\centering
\begin{tikzpicture}
\node[lat] (A) at (0,2.2)   {$A$};
\node[lat] (B) at (1.5,2.2) {$B$};
\node[obs] (D) at (3.0,2.2) {$D$};
\node[obs] (C) at (0.75,1.1){$C$};
\node[obs] (E) at (2.25,1.1){$E$};
\node[obs] (F) at (1.5,0)   {$F$};
\draw[ed] (A) -- (C);
\draw[ed] (B) -- (C);
\draw[ed] (B) -- (E);
\draw[ed] (D) -- (E);
\draw[ed] (C) -- (F);
\draw[ed] (E) -- (F);
\draw[ed] (A) to[bend right=40] (F);
\draw[ed, draw=red] (D) to[bend left=35] (F);
\node at (1.5,-0.85) {(a)};
\node at (5.5,-0.85) {(b)};
\begin{scope}[xshift=4cm]
\node[obs] (D) at (3.0,2.2) {$D$};
\node[lat] (C) at (1,1.2){$\Lambda$};
\node[obs] (E) at (2.25,1.1){$E$};
\node[obs] (F) at (1.5,0)   {$F$};
\draw[ed] (D) -- (E);
\draw[ed] (C) -- (F);
\draw[ed] (E) -- (F);
\draw[ed] (C) -- (E);
\end{scope}
\end{tikzpicture}
\caption{(a) The causal structures $H_1$ (without red arrow) and $G_1$ (with red arrow). 
(b) The instrumental causal structure which is used in the proof of our results.}
\label{fig:spurious1}
\end{figure}

Let $H$ be a causal structure with nodes $X_1,\ldots,X_n$ of which $X_{m+1},\ldots,X_n$ are latent and let $X_i^{\downarrow}$ denote the parents of $X_i$ in $H$ (i.e., the nodes that have a directed edge to $X_i$ in $H$). In classical theory, every latent node has an associated random variable, and a distribution $P(X_1\ldots X_m)$ over the observed variables is said to be \emph{classically compatible} with $H$ if there exists a distribution $Q(X_1\ldots X_n)$ over all the nodes such that
\begin{eqnarray}
  P(x_1\ldots x_m)&=&\sum_{x_{m+1}\ldots x_n}Q(x_1\ldots x_n),\text{ and}\\
  Q(x_1\ldots x_n)&=&\prod_{i=1}^nQ(x_i|x_i^{\downarrow})\,.\label{eq:prod}
\end{eqnarray}
This condition implies that any variable is conditionally independent of its non-descendants given its parents (Theorem~1.2.7 in~\cite{pearl}), which is known as the causal Markov condition.

For example, $P(CDEF)$ is classically compatible with the causal structure $H_1$ in Figure~\ref{fig:spurious1}(a) if there exists a joint distribution $Q(ABCDEF)$ such that
  \begin{multline}   
    P(cdef)=\sum_{ab}Q(abcdef) \text{\;\;and}\\
    Q(abcdef)=Q(a)Q(b)Q(d)Q(c|ab)Q(e|bd)Q(f|ace)\,.
     \label{eqn1}
	\end{multline}

In quantum theory, each of the latent nodes is associated with a quantum state instead of a random variable. In this work we only consider causal structures in which every latent node is parentless. Given a causal structure $H$, the set of quantum mechanically compatible distributions is defined as follows. Consider a latent node $X$ with children $X^\uparrow$. For each element $Y$ of $X^\uparrow$ associate a Hilbert space $\cH_X^Y$, and then associate a joint quantum state on the tensor product of these Hilbert spaces, $\bigotimes_{Y\in X^\uparrow}\cH_X^Y$, with $X$. For each observed node $Y$ with latent parents $Y^L$ and observed parents $Y^O$, associate a POVM that acts on $\bigotimes_{Q_i\in Y^L}\cH_{Q_i}^Y$, where the POVM can depend on the values of the variables $Y^O$. The joint distribution is formed by applying the POVMs to the quantum state. If $Y^L$ is empty, taking the empty tensor product to be $\mathbb{C}$, since a POVM on $\mathbb{C}$ is just a probability distribution, an observed node with no latent parents is assigned a probability distribution (conditioned on the values of any observed parents), recovering the classical case. Correlations obtained in such a way are called \emph{quantum mechanically compatible} with the causal structure.

For example, a distribution $P(CDEF)$ is quantum mechanically compatible with $H_1$ if there exist Hilbert spaces $\cH_A^F$, $\cH_A^C$, $\cH_B^C$, $\cH_B^E$, density operators $\rho_A$ on $\cH_A^F\ot\cH_A^C$ and $\rho_B$ on $\cH_B^C\ot\cH_B^E$, and POVMs $\{K^e_d\}_e$ on $\cH_B^E$ (for each $d$), $\{M^{f}_{c,e}\}_f$ on $\cH_A^F$ (for each $c,e$) and $\{N^c\}_c$ on $\cH_A^C\ot\cH_B^C$ such that
\begin{multline}
   P(cdef)=\tr((M^{f}_{c,e}\ot N^c\ot K^e_d)(\rho_A\ot\rho_B))P(d).
     \label{eqn2}
\end{multline}

Given a causal structure, certain conditional independences that hold for all (classically or quantum mechanically) compatible distributions can be directly read off its associated DAG using the concept of \(d\)-separation.

\begin{definition}[Blocked paths]
Consider a causal structure with $X$, $Y$ and $Z$ disjoint sets of nodes. A path from $X$ to $Y$ is \emph{blocked} by $Z$ if it contains either $A\to W\to B$ with $W\in Z$, $A\leftarrow W\to B$ with $W\in Z$ or $A\to W \leftarrow B$ such that neither $W$ nor any descendant of $W$ belongs to $Z$.
\end{definition}

\begin{definition}[\(d\)-separation]
\label{def: d-sep}
Consider a causal structure with disjoint sets of nodes, $X$, $Y$ and $Z$.  $X$ and $Y$ are \emph{d-separated} by $Z$, denoted $X\bot Y| Z$ if either there is no path between any node in $X$ and any node in $Y$, or if every path from a node in $X$ to a node in $Y$ is \emph{blocked} by $Z$. 
\end{definition}
In~\cite{Geiger1987,Verma1988} it was shown that if $X\bot Y| Z$ in a causal structure $H$, then any distribution $P$ that is classically compatible with $H$ satisfies $P(X|YZ)=P(X|Z)$, which we denote as $X\dbot Y| Z$. This was extended to quantum causal structures in~\cite{HLP_2014}. However, the converse does not hold and cases where it does not are said to be fine tuned (see Definition~\ref{fine tuning}).

It is further useful to consider distributions that are potentially beyond the quantum set. Given a causal structure $H$, a distribution $P$ over the observed nodes of $H$ is said to be \Icomp{} with $H$ if $P$ obeys all the conditional independences that follow from the \(d\)-separation relations on the observed nodes, i.e., if $X\dbot Y| Z$ for all sets of observed nodes $X$, $Y$, $Z$ for which $X\bot Y| Z$. We denote by $\mathcal{I}_H$ the set of all such probability distributions. The set of distributions classically compatible with $H$ is denoted as $\mathcal{C}_H$, meanwhile the set of quantum mechanically compatible correlations is $\mathcal{Q}_H$. From~\cite{HLP_2014}, we know that $\mathcal{C}_H \subseteq \mathcal{Q}_H \subseteq \mathcal{I}_H$. If $\mathcal{C}_H \subset \mathcal{I}_H$, then $H$ is said to support non-classical correlations. Further, if $\mathcal{C}_H \subset \mathcal{Q}_H$, then $H$ supports non-classical quantum correlations.

We now define what it means for correlations to be spurious. To do so, we formally introduce fine tuning and the set of non-classical non-fine tuned correlations in a causal structure.
\begin{definition}[Fine tuning]\label{fine tuning}
Consider a causal structure $H$ and a distribution $P$ over the observed nodes of $H$ that is \Icomp{} with $H$. If there exist disjoint sets of observed nodes $X$, $Y$ and $Z$ for which $P(X|YZ)=P(X|Z)$ but where the \(d\)-separation $X\bot Y| Z$ does not hold in $H$, then $P$ is said to be \emph{fine tuned} (with respect to $H$). Otherwise $P$ is said to be \emph{faithful} (with respect to $H$).
\end{definition}

\begin{definition}
    Let $H$ be a causal structure.  A distribution $P$ over the observed variables of $H$ is \emph{non-classical and non-fine tuned (NCNFT) in $H$} if
  \begin{itemize}
  \item $P$ is \Icomp{} with $H$;
  \item $P$ is not classically compatible with $H$; and
  \item $P$ is not fine tuned with respect to $H$.
  \end{itemize}
\end{definition}
\begin{definition}[Spurious correlations]
  Let $H$ be a causal structure and $P$ be NCNFT in $H$.
  $P$ is \emph{spurious with respect to $H$} if there exists a causal structure $G$ such that $P$ is classically compatible with $G$ and not fine tuned with respect to $G$. If, in addition, $P$ is quantum mechanically compatible with $H$, then $P$ is a spurious quantum correlation with respect to $H$.
\label{spur}
\end{definition}

Definition~\ref{spur} is useful in cases where we do not have any information about the spacetime locations at which particular variables are realised. If such information is available, it is useful to use the upgraded definition given below. To introduce this, we use $\bef_H$ to denote the partial order of the causal structure $H$, i.e., for two nodes $X$ and $Y$ in $H$, the relation $X\bef_H Y$ (``$X$ is before $Y$ in $H$'') means that there is a directed path from $X$ to $Y$ in $H$. We also use $X\nbef\naft_HY$ (``$X$ and $Y$ are not ordered in $H$'') to mean that there is no directed path from $X$ to $Y$ or from $Y$ to $X$ in $H$.

\begin{definition}
Let $H$ be a causal structure. A \emph{spacetime order} for $H$ is a partial order $\bef_{\st}$ on the observed nodes of $H$ that respects the partial order of $H$ in the sense that for every pair of observed nodes $X$ and $Y$ in $H$, \[X\bef_HY\implies X\bef_{\st}Y.\]
\end{definition}
Note that a spacetime order for $H$ need not be identical to the partial order of $H$; we can have $X\nbef\naft_HY$ but $X\bef_\st Y$ for instance.

\begin{definition}[Spurious correlations with a spacetime order]
  Let $H$ be a causal structure and let $\bef_{\st}$ be a spacetime order for $H$. A distribution $P$ over the observed variables of $H$ is said to be \emph{spurious with respect to $H$ with spacetime order $\bef_{\st}$} if $P$ satisfies all the conditions of Definition~\ref{spur} but with the additional requirement on $G$ that $\bef_\st$ is a spacetime order for $G$.
\label{spur+order}
\end{definition}
\begin{definition}[Spurious correlations respecting the causal partial order]
  Let $H$ be a causal structure. A distribution $P$ over the observed variables of $H$ is said to be \emph{spurious and respect the causal partial order of $H$} if $P$ satisfies all the conditions of Definition~\ref{spur} with the additional requirement on $G$ that $\bef_G$ is identical to $\bef_H$.
\label{spur+cpa}
\end{definition}
Definition~\ref{spur+cpa} has the advantage that it does not need an additional spacetime order to be defined. However, it is a stronger requirement than Definition~\ref{spur+order}, and may not be relevant in all cases. Its significance is that if $P$ is spurious and respects the causal partial order of $H$, then, for any spacetime order $\bef_\st$ for $H$, $P$ is spurious with respect to $H$ with spacetime order $\bef_\st$. This means that, subject to natural relativistic constraints, for all possible locations of the variables in spacetime there is an alternative classical causal structure that gives rise to $P$.

For all the instances discussed in this paper we prove spuriousness respecting the causal partial order.

\section{Results}
Our main aim is to establish the existence of quantum correlations that are spurious in $H_1$. We do so by showing that all quantum correlations in $H_1$ are classical in $G_1$.

We start by giving a condition that can identify spurious correlations.
\begin{theorem}
Let $H$ and $G$ be causal structures with the same set of observed nodes with $\bef_G$ identical to $\bef_H$, and satisfying $\mathcal{I}_H=\mathcal{I}_G$, $\mathcal{C}_H\subset \mathcal{I}_H$, $\mathcal{C}_G=\mathcal{I}_G$. Then every distribution that is NCNFT in $H$ is spurious and respects the causal partial order of $H$.
\label{th1}
\end{theorem}
\begin{proof}
From the conditions in the statement, $\mathcal{C}_G=\mathcal{I}_H$, from which it follows that any distribution that is \Icomp{} with $H$ is classically compatible with $G$. Further, since $\mathcal{I}_G=\mathcal{I}_H$, if the distribution is not fine tuned with respect to $H$ it is also not fine tuned with respect to $G$.
\end{proof}
The next step is to establish that $G_1$ has no non-classical correlations (i.e., $\mathcal{C}_{G_1}=\mathcal{I}_{G_1}$). To do so we use the following theorem, which is (part of) Theorem~26 of~\cite{HLP_2014} (Theorem~26 in~\cite{HLP_2014} has an additional transformation that we do not need here). 
\begin{theorem}
Consider a causal structure $G$ and let $G'$ be a second causal structure obtained by performing one of the following transformations on $G$: 1) Removing an edge, 2) Removing an isolated latent node, or 3) Adding an edge from a node $X$ to node $Y$ the parents of $X$ are a subset of the parents of $Y$, and the parents of $X$ contain at least one latent node. Then $\mathcal{C}_{G'}\subseteq\mathcal{C}_{G}$.
\label{th2}
\end{theorem}
\begin{lemma}[$\mathcal{C}_{G_1}=\mathcal{I}_{G_1}$]
A distribution is I-compatible with $G_1$ if and only if it is classically compatible with $G_1$.
\label{lm4}
\end{lemma}
\begin{proof}
We apply Theorem~\ref{th2} several times with the following sequence: remove the edges $A\rightarrow F$ and $A\rightarrow C$, remove the isolated latent node $A$, add the edge $C\rightarrow E$, remove the edges $B\rightarrow C$ and $B\rightarrow E$, remove the isolated latent node $B$. Call the resulting causal structure $G_1'$. Theorem~\ref{th2} implies that $\mathcal{C}_{G_1'}\subseteq\mathcal{C}_{G_1}$. However, since $G_1'$ has no latent nodes, it follows from~\cite{HLP_2014} that $\mathcal{C}_{G_1'}=\mathcal{I}_{G_1'}$. Further, since the \(d\)-separation relations between observed nodes are the same in both $G_1$ and $G_1'$, we have $\mathcal{I}_{G_1'}=\mathcal{I}_{G_1}$. Thus $\mathcal{I}_{G_1}\subseteq\mathcal{C}_{G_1}$, but since $\mathcal{C}_{G_1}\subseteq\mathcal{I}_{G_1}$ the claim follows. 
\end{proof}
\begin{theorem}
All NCNFT correlations in $H_1$ are spurious respecting the causal partial order of $H_1$.
\label{th3}
\end{theorem}
\begin{proof}
Consider the causal structure $G_1$ in Figure~\ref{fig:spurious1}(a) for which $\bef_{G_1}$ is identical to $\bef_{H_1}$. From Lemma~\ref{lm4} we have $\mathcal{C}_{G_1}=\mathcal{I}_{G_1}$. The only conditional independence on the observed nodes that follows from the \(d\)-separation relations is $C\indep D$, so $\mathcal{I}_{H_1}=\mathcal{I}_{G_1}$. Therefore, by Theorem~\ref{th1} all NCNFT correlations in $H_1$ are spurious respecting the causal partial order of $H_1$.
\end{proof}

Given Theorem~\ref{th3} it remains to find quantum NCNFT correlations in $H_1$. We proceed to construct a concrete example. The idea behind our construction is that if entanglement is shared by $A$ and $B$, and a Bell basis measurement performed to produce $C$ then the situation is analogous to having $\Lambda$ share entanglement in the instrumental causal structure (Figure~\ref{fig:spurious1}(b)), where the resulting entangled state depends on $C$. We then exploit a known non-classical quantum correlation in the instrumental causal structure from~\cite{chaves2018quantum,van2019quantum} to obtain the violation.

The next lemma gives the concrete reduction from $H_1$ to the instrumental causal structure. The general idea behind this has been called the marginalization piggyback in~\cite{Marinathesis}.
\begin{lemma}
Let $P(CDEF)$ be classically compatible with $H_1$ (cf.\ Figure~\ref{fig:spurious1}(a)). Then $P(DEF)$ is classically compatible with the instrumental causal structure (cf.\ Figure~\ref{fig:spurious1}(b)).
\label{lm1}    
\end{lemma}
\begin{proof}
From~\eqref{eqn1}, if $P$ is classically compatible with $H_1$ then
\begin{equation*}
P(def)=\sum_{abc}Q(a)\,Q(b)\,Q(d)\,Q(c| ab)\,Q(e| bd)\,Q(f| ace).
\end{equation*}
Define a new random variable $\Lambda$ with $Q(\lambda=(a,b,c))=Q(a)Q(b)Q(c|ab)$. Noting that $Q(e|bd)=Q(e|abcd)$ and $Q(f|ace)=Q(f|abce)$, we have
\begin{align*}
P(def)=\sum_\lambda Q(\lambda)Q(d)Q(e|\lambda d)Q(f|\lambda e),
\end{align*}
which is the condition required for $P$ to be classical in the instrumental causal structure. 
\end{proof}
We will use the following result, originally proven in~\cite{bonet:instruments}. 
\begin{lemma}
Let $P(DEF)$ be classically compatible with the instrumental causal structure, where the random variables take values $d\in\{0,1,2\}$, $e\in\{0,1\}$ and $f\in\{0,1\}$. Then
\begin{align}\label{eq:Bonet}
P(E{=}F|D{=}0)+P(F{=}0|D{=}1)+P(E{=}0,F{=}1|D{=}2)\!\le\!2.
\end{align}
\label{lm2}  
\end{lemma}
To construct our quantum distribution we use the notation $\ket{\theta}=\cos(\theta)\ket{0}+\sin(\theta)\ket{1}$, $\ket{\psi_x}=\cos(x)\ket{00}+\sin(x)\ket{11}$,
\begin{align*}
\ket{m_0}&=\tfrac{1}{\sqrt{2}}\big(\ket{00}+\ket{11}\big),& \ket{m_1}&=\tfrac{1}{\sqrt{2}}\big(\ket{01}+\ket{10}\big),\\ \ket{m_2}&=\tfrac{1}{\sqrt{2}}\big(\ket{01}-\ket{10}\big),&
\ket{m_3}&=\tfrac{1}{\sqrt{2}}\big(\ket{00}-\ket{11}\big).
\end{align*}
In the notation introduced before~\eqref{eqn2}, take the Hilbert spaces $\cH_A^F$, $\cH_A^C$, $\cH_B^C$, $\cH_B^E$ to be two dimensional, and the observed variables to take values $c\in\{0,1,2,3\}$, $d\in\{0,1,2\}$, $e\in\{0,1\}$ and $f\in\{0,1\}$. Then take
\begin{align}
  &\rho_A=\ketbra{\psi_\alpha}{\psi_\alpha},\ \rho_B=\ketbra{\psi_\beta}{\psi_\beta}\label{eq:q1}\\
  &P(d)=1/3,\ N^c=\ketbra{m_c}{m_c}\\
&K^e_d=\ketbra{\epsilon_d + e\tfrac{\pi}{2}}{\epsilon_d + e\tfrac{\pi}{2}}\ \text{ with }\ (\epsilon_0,\epsilon_1,\epsilon_2)=\big(\tfrac{\pi}{4},\,0,\,\tfrac{5\pi}{8}\big)\\
  &M^f_{c,e}=\sigma_c\ketbra{\theta^f_e}{\theta^f_e}\sigma_c^\dagger\ \text{ with }\ \theta^f_e=\frac{\pi}{8}+\phi+e\,\frac{\pi}{4}+f\,\frac{\pi}{2}.  \label{eq:q4}
\end{align}
Here $\{\sigma_c\}_{c=1}^3$ are the Pauli operators and $\sigma_0=\id$, and the form of some of the states and measurements takes inspiration from the result of~\cite{chaves2018quantum,van2019quantum}.                                                                         
\begin{lemma}
  The distribution $P(CDEF)$ generated by the quantum strategy from~\eqref{eq:q1}--\eqref{eq:q4} is NCNFT in $H_1$.
\label{NCNFT_proof}  
\end{lemma}
\begin{proof}
Let $\chi_c$ be the density operator on $\cH_A^F\ot\cH_B^E$ conditioned on $C=c$. We have
\begin{align*}
  \chi_0&=\ketbra{\psi_{\zeta_0}}{\psi_{\zeta_0}}\\
  \chi_1&=(\sigma_1\ot\id)\ketbra{\psi_{\zeta_1}}{\psi_{\zeta_1}}(\sigma_1\ot\id)\\
  \chi_2&=(\sigma_2\ot\id)\ketbra{\psi_{\zeta_1}}{\psi_{\zeta_1}}(\sigma_2\ot\id)\\
  \chi_3&=(\sigma_3\ot\id)\ketbra{\psi_{\zeta_0}}{\psi_{\zeta_0}}(\sigma_3\ot\id),\ \text{ where}\\
          \tan(\zeta_0)&=\tan(\alpha)\tan(\beta)\text{ and }\tan(\zeta_1)=\cot(\alpha)\tan(\beta).
\end{align*}
The distribution over $C$ is $P(C=0)=P(C=3)=\tfrac{1}{4}\left(1+\cos(2\alpha)\cos(2\beta)\right)$ and $P(C=1)=P(C=2)=\tfrac{1}{4}\left(1-\cos(2\alpha)\cos(2\beta)\right)$.

We proceed to compute the value of the left-hand-side of~\eqref{eq:Bonet}.
In the case $C=0$, we have
\begin{align}
  P&(E=F|D=0,C=0)=\sum_e\tr((M^e_{0,e}\ot K^e_0)\chi_0)\nonumber\\
     &=\frac{1}{4}\left(2+\sin(2\zeta_0+\tfrac{\pi}{4}+2\phi)+\sin(2\zeta_0+\tfrac{\pi}{4}-2\phi)\right)\nonumber\\
  &=\frac{1}{2}\left(1+\cos(2\phi)\sin(2\zeta_0+\tfrac{\pi}{4})\right)\label{p1}
\end{align}
Similarly,
\begin{align}
  P&(F=0|D=1,C=0)=\sum_e\tr((M^0_{0,e}\ot K^e_1)\chi_0)\nonumber\\
  &=\cos^2(\phi+\tfrac{\pi}{8})\cos^2(\zeta_0)+\cos^2(\phi-\tfrac{\pi}{8})\sin^2(\zeta_0)\label{p2}
\end{align}
and
\begin{align}
  &P(E=0,F=1|D=2,C=0)=\tr((M^1_{0,0}\ot K^0_2)\chi_0)\nonumber\\
  &=\left(\cos(\zeta_0)\sin(\phi+\tfrac{\pi}{8})\sin(\tfrac{\pi}{8})+\sin(\zeta_0)\cos(\phi+\tfrac{\pi}{8})\cos(\tfrac{\pi}{8})\right)^2\label{p3}
\end{align}
Define the sum of the right-hand-sides of~\eqref{p1},~\eqref{p2} and~\eqref{p3} to be $I(\zeta_0,\phi)$. By construction, in the case $C=3$ we also get $I(\zeta_0,\phi)$, while for $C=1$ and $C=2$ we get $I(\zeta_1,\phi)$ (for each $c$, the Pauli operators in the definition of $M^f_{c,e}$ cancel those in $\chi_c$).

Note that
\[I(\zeta_0,0)=\frac{1}{8}\left(11+2\sqrt{2}+(1+2\sqrt{2})\sin(2\zeta_0)\right),\]
which is greater than the classical bound of $2$ for $\sin(2\zeta_0)>(5-2\sqrt2)/(1+2\sqrt2)\approx0.567$, with a maximum value of $(3+\sqrt2)/2\approx2.207$ at $\zeta_0=\pi/4$. However, we cannot take $\phi=0$ because then $P(F|CD)=P(F|D)$, which would correspond to a fine tuning. We hence choose $\phi$ to be small but non-zero to avoid this. Similarly, taking $\alpha=\pi/4$ or $\beta=\pi/4$ (which would correspond to the measurement generating $C$ doing perfect teleportation) leads to fine tuning ($C\dbot E$), so we choose values close to these instead.

Concretely, for $\alpha=7\pi/32$, $\beta=6\pi/32$, and $\phi=\pi/32$ we find the left-hand-side of~\eqref{eq:Bonet} evaluates to
\begin{align*}
P(C\in&\{0,3\})I(\zeta_0,\phi)+P(C\in\{1,2\})I(\zeta_1,\phi)\\
  =\frac{1}{32}\Big(&43+3\sqrt{2}-4\sin(\tfrac{\pi
   }{16})+8\sin(\tfrac{\pi}{8})+4\sin\tfrac{3\pi}{16})\\
  &+8\cos(\tfrac{\pi}{16})+2\cos(\tfrac{\pi}{8})+8\cos(\tfrac{3\pi}{16})\Big)\\
  \approx2.128
\end{align*}
Thus, by Lemma~\ref{lm2}, $P(DEF)$ is not classically compatible with the instrumental scenario, and so, by Lemma~\ref{lm1}, $P(CDEF)$ is not classically compatible with $H_1$. An exhaustive check of all conditional independence relations involving subsets of $C$, $D$, $E$ and $F$ confirms that $C\indep D$ is the only one satisfied by $P(CDEF)$. Thus, $P(CDEF)$ is NCNFT in $H_1$.
\end{proof}

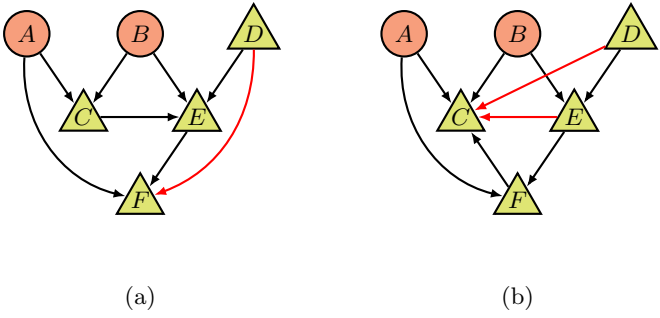
\begin{figure}[t]
\centering
\begin{tikzpicture}
\begin{scope}[xshift=-0.84cm]
\node[lat] (A2) at (0,2.2)   {$A$};
\node[lat] (B2) at (1.5,2.2) {$B$};
\node[obs] (D2) at (3.0,2.2) {$D$};
\node[obs] (C2) at (0.75,1.1){$C$};
\node[obs] (E2) at (2.25,1.1){$E$};
\node[obs] (F2) at (1.5,0)   {$F$};
\draw[ed] (A2) -- (C2);
\draw[ed] (B2) -- (C2);
\draw[ed] (B2) -- (E2);
\draw[ed] (D2) -- (E2);
\draw[ed] (C2) -- (E2);
\draw[ed] (E2) -- (F2);
\draw[ed] (A2) to[bend right=40] (F2);
\draw[ed, draw=red] (D2) to[bend left=35] (F2);
\node at (1.5,-1.285) {(a)};
\end{scope}
\begin{scope}[xshift=4.15cm]
\node[lat] (A2) at (0,2.2)   {$A$};
\node[lat] (B2) at (1.5,2.2) {$B$};
\node[obs] (D2) at (3.0,2.2) {$D$};
\node[obs] (C2) at (0.75,1.1){$C$};
\node[obs] (E2) at (2.25,1.1){$E$};
\node[obs] (F2) at (1.5,0)   {$F$};
\draw[ed] (A2) -- (C2);
\draw[ed] (B2) -- (C2);
\draw[ed] (B2) -- (E2);
\draw[ed] (D2) -- (E2);
\draw[ed] (F2) -- (C2);
\draw[ed] (E2) -- (F2);
\draw[ed] (A2) to[bend right=40] (F2);
\draw[ed, draw=red] (D2) -- (C2);
\draw[ed, draw=red] (E2) -- (C2);
\node at (1.5,-1.285) {(b)};
\end{scope}
\end{tikzpicture}
\caption{(a) The causal structures $H_2$ (without red arrow) and $G_2$ (with red arrow). (b) The causal structures $H_3$ (without red arrows) and $G_3$ (with red arrows).}
\label{fig:spurious2}
\end{figure}

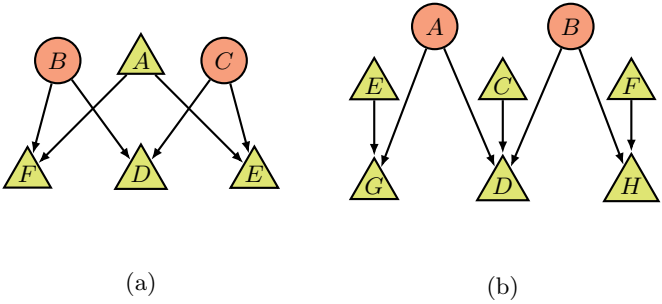
\begin{figure}[t]
\begin{tikzpicture}
\begin{scope}[yshift=-0.45cm]
\node[lat] (B1) at (0.6,2.1) {$B$};
\node[obs] (A1) at (1.7,2.1) {$A$};
\node[lat] (C1) at (2.8,2.1) {$C$};
\node[obs] (F1) at (0.2,0.6) {$F$};
\node[obs] (D1) at (1.7,0.6) {$D$};
\node[obs] (E1) at (3.2,0.6) {$E$};
\draw[ed] (B1) -- (F1); \draw[ed] (B1) -- (D1);
\draw[ed] (C1) -- (D1); \draw[ed] (C1) -- (E1);
\draw[ed] (A1) -- (F1); %controls (1.1,3.7) and (-0.3,2.8) .. ;
\draw[ed] (A1) -- (E1); %controls (2.3,3.7) and (3.7,2.8) .. ;
\node at (1.7,-0.85) {(a)};
\end{scope}
\begin{scope}[xshift=4.785cm]
\node[lat] (L1) at (0.8,2.1) {$A$};
\node[lat] (L2) at (2.6,2.1) {$B$};
\node[obs] (x) at (0,1.32)   {$E$};
\node[obs] (y) at (1.7,1.32) {$C$};
\node[obs] (z) at (3.4,1.32) {$F$};
\node[obs] (a) at (0,0)   {$G$};
\node[obs] (b) at (1.7,0) {$D$};
\node[obs] (c) at (3.4,0) {$H$};
\draw[ed] (L1) -- (a); \draw[ed] (L1) -- (b);
\draw[ed] (L2) -- (b); \draw[ed] (L2) -- (c);
\draw[ed] (x) -- (a); \draw[ed] (y) -- (b); \draw[ed] (z) -- (c);
\node at (1.7,-1.35) {(b)};
\end{scope}
\end{tikzpicture}
\caption{Causal structures (a) $H_4$ and (b) Bi-locality in which all NCNFT correlations are not spurious.}
\label{fig:genuine}
\end{figure}
Using similar reasoning, we find spurious non-classical (quantum) correlations respecting the causal partial order in the causal structures $H_2$ and $H_3$ in Figures~\ref{fig:spurious2}(a) and~\ref{fig:spurious2}(b). The corresponding classically compatible explanations are provided by $G_2$ and $G_3$ respectively (shown in the same figures). This can be proven using Theorem~\ref{th1} (we omit the proof as it follows the same line as the case of $H_1$, previously discussed). The triangle and the instrumental causal structures also exhibit spurious quantum correlations respecting the causal partial order, since any non-classical quantum correlations in them is classically compatible with the causal structure where all the observed variables share a single common cause.

We also find that for some other causal structures (shown in Figure~\ref{fig:genuine}) there are NCNFT correlations that are not spurious.
\begin{theorem}
Let $H$ be one of the causal structures in Figure~\ref{fig:genuine} and $P$ be NCNFT in $H$. $P$ is not spurious.
\label{th4}
\end{theorem}
The idea behind the proof is to argue that common causes of three or more observed variables are incompatible with the \(d\)-separation relations between the observed nodes (without fine tuning), and then to use IC* to rule out alternative causal structures with at most bipartite common causes.
\begin{proof}
For $H_4$ the \(d\)-separation relations involving only observed nodes are $\{A \bot D, E \bot F | A\}$. Any common cause of three (or more) observed nodes will be a common cause for either $A$ and $D$ or $E$ and $F$. However, since $A\bot D$ and $E\bot F|A$, this will lead to fine tuning. Thus, faithful correlations that are \Icomp{} with $H_4$ can have at most bi-partite common causes. For correlations with $A \dbot D$ and $E \dbot F | A$ (and no other conditional independences), the IC* algorithm returns a set of causal structures all of which are reducible to $H_4$ using the techniques of~\cite{HLP_2014}. Thus, NCNFT correlations in $H_4$ cannot be spurious.

For the bi-locality causal structure, the \(d\)-separation relations involving only observed nodes are generated by $\{E\bot CDFH,\;F\bot CDEG,\;C\bot EGFH,\; CFH\bot EG,\; G\bot H\}$ (via the semi-graphoid axioms~\cite{pearl2022graphoids}). From these \(d\)-separations it follows that a common cause of $\{E,C\}$, $\{E,F\}$, $\{C,F\}$, $\{E,D\}$, $\{E,H\}$, $\{C,G\}$, $\{C,H\}$, $\{F,G\}$ or $\{F,D\}$ will require fine tuning. This leaves the only possible common cause of more than 2 observed variables as a common cause of $\{G,D,H\}$. However, with such a common cause fine tuning would be needed to satisfy $G\dbot H$. Hence only bi-partite common causes can avoid fine tuning. For correlations with the conditional independences implied by the \(d\)-separation relations above, the IC* algorithm returns a set of causal structures all of which are reducible to the bi-locality structure using the techniques of~\cite{HLP_2014}.

Hence, for both the causal structures in Figure~\ref{fig:genuine}, since there do not exist any other alternate causal structures that exhibit the same corresponding conditional independences, we have that all NCNFT correlations in them are not spurious.
\end{proof}

\section{Conclusion}
We have given several examples of causal structures that have spurious quantum correlations, and several others that do not. In the spurious cases, a conclusion that quantum correlations are present could be due to using the wrong causal structure, with a classical explanation being possible in another. The quantum nature of such correlations would hence not be inferrable by causal inference alone but would require additional analysis. Furthermore, the alternative classical explanation would remain compatible with relativity theory based on the spacetime locations where the variables are generated, so compatibility with relativity can also not help eliminate the classical alternative.

In the cases without spurious correlations, we get a more direct reason to believe in the quantum nature of particular observed correlations: these correlations do not have a (non fine tuned) classical explanation in any other causal structure. These non-spurious correlations are hence natural candidates for use in information-processing tasks, analogously to the use of non-classical correlations for device-independent cryptographic tasks such as quantum key distribution or randomness expansion in Bell causal structures (see, e.g.,~\cite{PUL&} for a review).

\subsection*{Acknowledgements}
We thank Robert Spekkens and Elie Wolfe for useful conversations.
%\bibliography{new_refs}

\begin{thebibliography}{14}%
\makeatletter
\providecommand \@ifxundefined [1]{%
 \@ifx{#1\undefined}
}%
\providecommand \@ifnum [1]{%
 \ifnum #1\expandafter \@firstoftwo
 \else \expandafter \@secondoftwo
 \fi
}%
\providecommand \@ifx [1]{%
 \ifx #1\expandafter \@firstoftwo
 \else \expandafter \@secondoftwo
 \fi
}%
\providecommand \natexlab [1]{#1}%
\providecommand \enquote  [1]{``#1''}%
\providecommand \bibnamefont  [1]{#1}%
\providecommand \bibfnamefont [1]{#1}%
\providecommand \citenamefont [1]{#1}%
\providecommand \href@noop [0]{\@secondoftwo}%
\providecommand \href [0]{\begingroup \@sanitize@url \@href}%
\providecommand \@href[1]{\@@startlink{#1}\@@href}%
\providecommand \@@href[1]{\endgroup#1\@@endlink}%
\providecommand \@sanitize@url [0]{\catcode `\\12\catcode `\$12\catcode
  `\&12\catcode `\#12\catcode `\^12\catcode `\_12\catcode `\%12\relax}%
\providecommand \@@startlink[1]{}%
\providecommand \@@endlink[0]{}%
\providecommand \url  [0]{\begingroup\@sanitize@url \@url }%
\providecommand \@url [1]{\endgroup\@href {#1}{\urlprefix }}%
\providecommand \urlprefix  [0]{URL }%
\providecommand \Eprint [0]{\href }%
\providecommand \doibase [0]{http://dx.doi.org/}%
\providecommand \selectlanguage [0]{\@gobble}%
\providecommand \bibinfo  [0]{\@secondoftwo}%
\providecommand \bibfield  [0]{\@secondoftwo}%
\providecommand \translation [1]{[#1]}%
\providecommand \BibitemOpen [0]{}%
\providecommand \bibitemStop [0]{}%
\providecommand \bibitemNoStop [0]{.\EOS\space}%
\providecommand \EOS [0]{\spacefactor3000\relax}%
\providecommand \BibitemShut  [1]{\csname bibitem#1\endcsname}%
\let\auto@bib@innerbib\@empty
%</preamble>
\bibitem [{\citenamefont {Bell}(1964)}]{bell}%
  \BibitemOpen
  \bibfield  {author} {\bibinfo {author} {\bibfnamefont {J.~S.}\ \bibnamefont
  {Bell}},\ }\bibfield  {title} {\enquote {\bibinfo {title} {{On the
  Einstein-Podolsky-Rosen paradox}},}\ }\href {\doibase
  10.1103/PhysicsPhysiqueFizika.1.195} {\bibfield  {journal} {\bibinfo
  {journal} {Physics}\ }\textbf {\bibinfo {volume} {1}},\ \bibinfo {pages}
  {195--200} (\bibinfo {year} {1964})}\BibitemShut {NoStop}%
\bibitem [{\citenamefont {Bell}(2004)}]{bellnouvelle}%
  \BibitemOpen
  \bibfield  {author} {\bibinfo {author} {\bibfnamefont {J.~S.}\ \bibnamefont
  {Bell}},\ }\bibfield  {title} {\enquote {\bibinfo {title} {La nouvelle
  cuisine},}\ }in\ \href {\doibase 10.1017/CBO9780511815676.026} {\emph
  {\bibinfo {booktitle} {Speakable and Unspeakable in quantum mechanics}}}\
  (\bibinfo  {publisher} {Cambridge University Press},\ \bibinfo {year}
  {2004})\ pp.\ \bibinfo {pages} {232--248}\BibitemShut {NoStop}%
\bibitem [{\citenamefont {Wood}\ and\ \citenamefont {Spekkens}(2015)}]{wood}%
  \BibitemOpen
  \bibfield  {author} {\bibinfo {author} {\bibfnamefont {C.~J.}\ \bibnamefont
  {Wood}}\ and\ \bibinfo {author} {\bibfnamefont {R.~W.}\ \bibnamefont
  {Spekkens}},\ }\bibfield  {title} {\enquote {\bibinfo {title} {The lesson of
  causal discovery algorithms for quantum correlations: Causal explanations of
  {B}ell-inequality violations require fine-tuning},}\ }\href {\doibase
  10.1088/1367-2630/17/3/033002} {\bibfield  {journal} {\bibinfo  {journal}
  {New J. Phys.}\ }\textbf {\bibinfo {volume} {17}},\ \bibinfo {pages} {033002}
  (\bibinfo {year} {2015})}\BibitemShut {NoStop}%
\bibitem [{\citenamefont {Pearl}(2009)}]{pearl}%
  \BibitemOpen
  \bibfield  {author} {\bibinfo {author} {\bibfnamefont {J.}~\bibnamefont
  {Pearl}},\ }\href {\doibase 10.1017/CBO9780511803161} {\emph {\bibinfo
  {title} {{Causality}}}},\ \bibinfo {edition} {2nd}\ ed.\ (\bibinfo
  {publisher} {Cambridge University Press},\ \bibinfo {year}
  {2009})\BibitemShut {NoStop}%
\bibitem [{\citenamefont {Verma}(1993)}]{verma1993graphical}%
  \BibitemOpen
  \bibfield  {author} {\bibinfo {author} {\bibfnamefont {T.}~\bibnamefont
  {Verma}},\ }\bibfield  {title} {\enquote {\bibinfo {title} {Graphical aspects
  of causal models},}\ }\href {https://ftp.cs.ucla.edu/pub/stat_ser/r191.pdf}
  {\bibfield  {journal} {\bibinfo  {journal} {Technical Report R-191, UCLA}\ }
  (\bibinfo {year} {1993})}\BibitemShut {NoStop}%
\bibitem [{\citenamefont {Henson}\ \emph {et~al.}(2014)\citenamefont {Henson},
  \citenamefont {Lal},\ and\ \citenamefont {Pusey}}]{HLP_2014}%
  \BibitemOpen
  \bibfield  {author} {\bibinfo {author} {\bibfnamefont {J.}~\bibnamefont
  {Henson}}, \bibinfo {author} {\bibfnamefont {R.}~\bibnamefont {Lal}}, \ and\
  \bibinfo {author} {\bibfnamefont {M.~F.}\ \bibnamefont {Pusey}},\ }\bibfield
  {title} {\enquote {\bibinfo {title} {{Theory-independent limits on
  correlations from generalized Bayesian networks}},}\ }\href {\doibase
  10.1088/1367-2630/16/11/113043} {\bibfield  {journal} {\bibinfo  {journal}
  {New J. Phys.}\ }\textbf {\bibinfo {volume} {16}},\ \bibinfo {pages} {113043}
  (\bibinfo {year} {2014})}\BibitemShut {NoStop}%
\bibitem [{\citenamefont {Geiger}(1987)}]{Geiger1987}%
  \BibitemOpen
  \bibfield  {author} {\bibinfo {author} {\bibfnamefont {D.}~\bibnamefont
  {Geiger}},\ }\href {http://fmdb.cs.ucla.edu/Treports/880053.pdf} {\emph
  {\bibinfo {title} {Towards the Formalization of Informational
  Dependencies}}},\ \bibinfo {type} {Tech. Rep.}\ \bibinfo {number} {880053}\
  (\bibinfo  {institution} {UCLA Computer Science},\ \bibinfo {year}
  {1987})\BibitemShut {NoStop}%
\bibitem [{\citenamefont {Verma}\ and\ \citenamefont
  {Pearl}(1988)}]{Verma1988}%
  \BibitemOpen
  \bibfield  {author} {\bibinfo {author} {\bibfnamefont {T.}~\bibnamefont
  {Verma}}\ and\ \bibinfo {author} {\bibfnamefont {J.}~\bibnamefont {Pearl}},\
  }\bibfield  {title} {\enquote {\bibinfo {title} {Causal networks: Semantics
  and expressiveness},}\ }in\ \href {\doibase 10.48550/arXiv.1304.2379} {\emph
  {\bibinfo {booktitle} {Proceedings of the 4th Workshop on Uncertainty in
  Artificial Intelligence}}}\ (\bibinfo {year} {1988})\ pp.\ \bibinfo {pages}
  {352--359}\BibitemShut {NoStop}%
\bibitem [{\citenamefont {Chaves}\ \emph {et~al.}(2018)\citenamefont {Chaves},
  \citenamefont {Carvacho}, \citenamefont {Agresti}, \citenamefont {Di~Giulio},
  \citenamefont {Aolita}, \citenamefont {Giacomini},\ and\ \citenamefont
  {Sciarrino}}]{chaves2018quantum}%
  \BibitemOpen
  \bibfield  {author} {\bibinfo {author} {\bibfnamefont {R.}~\bibnamefont
  {Chaves}}, \bibinfo {author} {\bibfnamefont {G.}~\bibnamefont {Carvacho}},
  \bibinfo {author} {\bibfnamefont {I.}~\bibnamefont {Agresti}}, \bibinfo
  {author} {\bibfnamefont {V.}~\bibnamefont {Di~Giulio}}, \bibinfo {author}
  {\bibfnamefont {L.}~\bibnamefont {Aolita}}, \bibinfo {author} {\bibfnamefont
  {S.}~\bibnamefont {Giacomini}}, \ and\ \bibinfo {author} {\bibfnamefont
  {F.}~\bibnamefont {Sciarrino}},\ }\bibfield  {title} {\enquote {\bibinfo
  {title} {Quantum violation of an instrumental test},}\ }\href
  {https://www.nature.com/articles/s41567-017-0008-5} {\bibfield  {journal}
  {\bibinfo  {journal} {Nature Physics}\ }\textbf {\bibinfo {volume} {14}},\
  \bibinfo {pages} {291--296} (\bibinfo {year} {2018})}\BibitemShut {NoStop}%
\bibitem [{\citenamefont {Van~Himbeeck}\ \emph {et~al.}(2019)\citenamefont
  {Van~Himbeeck}, \citenamefont {Brask}, \citenamefont {Pironio}, \citenamefont
  {Ramanathan}, \citenamefont {Sainz},\ and\ \citenamefont
  {Wolfe}}]{van2019quantum}%
  \BibitemOpen
  \bibfield  {author} {\bibinfo {author} {\bibfnamefont {T.}~\bibnamefont
  {Van~Himbeeck}}, \bibinfo {author} {\bibfnamefont {J.~B.}\ \bibnamefont
  {Brask}}, \bibinfo {author} {\bibfnamefont {S.}~\bibnamefont {Pironio}},
  \bibinfo {author} {\bibfnamefont {R.}~\bibnamefont {Ramanathan}}, \bibinfo
  {author} {\bibfnamefont {A.~B.}\ \bibnamefont {Sainz}}, \ and\ \bibinfo
  {author} {\bibfnamefont {E.}~\bibnamefont {Wolfe}},\ }\bibfield  {title}
  {\enquote {\bibinfo {title} {Quantum violations in the instrumental scenario
  and their relations to the {B}ell scenario},}\ }\href {\doibase
  10.22331/q-2019-09-16-186} {\bibfield  {journal} {\bibinfo  {journal}
  {Quantum}\ }\textbf {\bibinfo {volume} {3}},\ \bibinfo {pages} {186}
  (\bibinfo {year} {2019})}\BibitemShut {NoStop}%
\bibitem [{\citenamefont {Maciel~Ansanelli}(2026)}]{Marinathesis}%
  \BibitemOpen
  \bibfield  {author} {\bibinfo {author} {\bibfnamefont {M.}~\bibnamefont
  {Maciel~Ansanelli}},\ }\emph {\bibinfo {title} {Bidirectional insights
  between classical and quantum causal inference}},\ \href
  {https://uwspace.uwaterloo.ca/items/51f053aa-594b-42e1-b60d-80ceb465d984}
  {Ph.D. thesis},\ \bibinfo  {school} {University of Waterloo} (\bibinfo {year}
  {2026})\BibitemShut {NoStop}%
\bibitem [{\citenamefont {Bonet}(2001)}]{bonet:instruments}%
  \BibitemOpen
  \bibfield  {author} {\bibinfo {author} {\bibfnamefont {B.}~\bibnamefont
  {Bonet}},\ }\bibfield  {title} {\enquote {\bibinfo {title} {Instrumentality
  tests revisited},}\ }in\ \href {https://arxiv.org/abs/1301.2258} {\emph
  {\bibinfo {booktitle} {Proc. 17th Conf. on Uncertainty in Artificial
  Intelligence}}},\ \bibinfo {editor} {edited by\ \bibinfo {editor}
  {\bibfnamefont {J.}~\bibnamefont {Breese}}\ and\ \bibinfo {editor}
  {\bibfnamefont {D.}~\bibnamefont {Koller}}}\ (\bibinfo  {publisher} {Morgan
  Kaufmann},\ \bibinfo {address} {Seattle, WA},\ \bibinfo {year} {2001})\ pp.\
  \bibinfo {pages} {48--55}\BibitemShut {NoStop}%
\bibitem [{\citenamefont {Pearl}\ and\ \citenamefont
  {Paz}(2022)}]{pearl2022graphoids}%
  \BibitemOpen
  \bibfield  {author} {\bibinfo {author} {\bibfnamefont {J.}~\bibnamefont
  {Pearl}}\ and\ \bibinfo {author} {\bibfnamefont {A.}~\bibnamefont {Paz}},\
  }\bibfield  {title} {\enquote {\bibinfo {title} {Graphoids: Graph-based logic
  for reasoning about relevance relations or when would x tell you more about y
  if you already know z?}}\ }in\ \href {\doibase 10.1145/3501714.3501729}
  {\emph {\bibinfo {booktitle} {Probabilistic and Causal Inference: The Works
  of Judea Pearl}}}\ (\bibinfo {year} {2022})\ pp.\ \bibinfo {pages}
  {189--200}\BibitemShut {NoStop}%
\bibitem [{\citenamefont {Pirandola}\ \emph {et~al.}(2020)\citenamefont
  {Pirandola}, \citenamefont {Andersen}, \citenamefont {Banchi}, \citenamefont
  {Berta}, \citenamefont {Bunandar}, \citenamefont {Colbeck}, \citenamefont
  {Englund}, \citenamefont {Gehring}, \citenamefont {Lupo}, \citenamefont
  {Ottaviani} \emph {et~al.}}]{PUL&}%
  \BibitemOpen
  \bibfield  {author} {\bibinfo {author} {\bibfnamefont {S.}~\bibnamefont
  {Pirandola}}, \bibinfo {author} {\bibfnamefont {U.~L.}\ \bibnamefont
  {Andersen}}, \bibinfo {author} {\bibfnamefont {L.}~\bibnamefont {Banchi}},
  \bibinfo {author} {\bibfnamefont {M.}~\bibnamefont {Berta}}, \bibinfo
  {author} {\bibfnamefont {D.}~\bibnamefont {Bunandar}}, \bibinfo {author}
  {\bibfnamefont {R.}~\bibnamefont {Colbeck}}, \bibinfo {author} {\bibfnamefont
  {D.}~\bibnamefont {Englund}}, \bibinfo {author} {\bibfnamefont
  {T.}~\bibnamefont {Gehring}}, \bibinfo {author} {\bibfnamefont
  {C.}~\bibnamefont {Lupo}}, \bibinfo {author} {\bibfnamefont {C.}~\bibnamefont
  {Ottaviani}},  \emph {et~al.},\ }\bibfield  {title} {\enquote {\bibinfo
  {title} {Advances in quantum cryptography},}\ }\href {\doibase
  10.1364/AOP.361502} {\bibfield  {journal} {\bibinfo  {journal} {Advances in
  Optics and Photonics}\ }\textbf {\bibinfo {volume} {12}},\ \bibinfo {pages}
  {1012} (\bibinfo {year} {2020})}\BibitemShut {NoStop}%
\end{thebibliography}
%

\end{document}